\documentclass[runningheads, 11pt]{llncs}

\usepackage{geometry}
\usepackage{comment}
\usepackage[utf8]{inputenc}

\usepackage{amsthm, amsmath, amssymb}
\usepackage{hyperref, natbib, comment}       
\usepackage{graphicx}
\usepackage{titling}
\usepackage{IEEEtrantools}

\usepackage{hyperref,enumitem}
\hypersetup{
	colorlinks=true,
	citecolor = blue       %
}
\DeclareMathOperator{\EX}{\mathbb{E}}

\DeclareMathOperator*{\argmax}{arg\,max}

\usepackage{color}

\newcount\Comments 
\newcommand{\kibitz}[2]{\ifnum\Comments=1{\color{#1}{#2}}\fi}

\newcommand{\dcp}[1]{\kibitz{red}{[DCP: #1]}}

\definecolor{ao(english)}{rgb}{0.0, 0.5, 0.0}

\newcount\CommentsAdd 
\newcommand{\kibitzAdd}[2]{\ifnum\CommentsAdd=1{\color{#1}{#2}}\fi}
\newcommand{\dcpadd}[1]{\kibitzAdd{black}{#1}}

\title{Eliciting Social Knowledge for Creditworthiness Assessment}
\author{Mark York \\ Harvard University \\ markyork@g.harvard.edu \and Munther Dahleh \\ MIT \\ dahleh@mit.edu \and David C. Parkes \\ Harvard University \\ parkes@eecs.harvard.edu}

\begin{document}

\maketitle

\begin{abstract}
    	Access   to  capital is a major constraint for economic growth in the developing world. Yet those attempting to lend in this space face high defaults due to their inability to distinguish creditworthy borrowers from the rest. In this paper, we propose two novel scoring mechanisms that incentivize community members to truthfully report their signal on the creditworthiness of others in their community. We first design a truncated asymmetric scoring-rule for a setting where the lender has no liquidity constraints. We then derive a novel, strictly-proper VCG scoring mechanism for the liquidity-constrained setting. Whereas \citet{ChenYilingDMwG} give an impossibility result for an analogous setting in which sequential reports are made in the context of decision markets, we achieve a positive result through appeal to interim beliefs about the reports of others in a setting with simultaneous reports. Moreover, the use of VCG methods allows for the integration of linear belief aggregation methods.
\end{abstract}

\section{Introduction}
Access to capital has become the primary anti-poverty tool in development. Global microfinance grew from 13 million borrowers and \$7 billion in loans in 1995 \citep{KassimSalinaH.JHdri} to 140 million borrowers and \$129 billion in loans in 2019 \citep{microcreditsummit.org}. A particular challenge with microfinance is that the unbanked have minimal credit history, creating an information asymmetry problem between lenders and borrowers.

Muhammad Yunus launched microfinance in 1976 with the Grameen bank. They lend to groups of people who are jointly-liable to repay the loan. This creates self-selection based on community information \citep{Grameen}, but it also imposes significant cost on lenders and borrowers through bi-weekly meetings, the risk of default by fellow group members, and administration.

Another solution is the advent of data-analytics based lenders. These lenders typically give loans to individuals, and they leverage demographic or other information to select borrowers. Branch, operating in Kenya, requires users to own a smartphone with their app installed and runs analytics on the calls, text messages, emails, and other usage data from the phone. Based on the performance of past borrowers, these companies determine how likely a new potential borrower is to repay. Loans are as small as five USD, and interest rates start at 18\% monthly (199\% APR) ~\citep{Branch}. While this expands credit access, it excludes people who do not have smartphones and the interest rates are high. Another issue is that un-creditworthy borrowers learn which factors the algorithm considers, and they can modify their behavior to receive loans (e.g. \citep{BjorkegrenDaniel2020BRiM}).

Fortunately, research shows that community members are knowledgeable about the creditworthiness of people in their community. \citet{Maitra} deployed an agent-intermediated lending scheme in West Bengal, India through which they appointed agents to select borrowers and administer loans. These agents were compensated based on repayment rates, and the repayment rates were higher than those for group lending schemes in the same region. \citet{Hussam18} went one step further and deployed a commmunity-recommendation scheme employing the {\em Robust Bayesian Truth Serum} (RBTS) to reward recommenders for giving reports that conform closely to those of their peers~\citep{WitkowskiP12}. RBTS was found to partially nullify the incentives of recommenders to lie on behalf of family members. Of note, RBTS does not reward recommenders based on repayment outcomes.

In this paper, we propose a new information elicitation system that incentivizes community members to report their true beliefs about the likelihood that others will repay a loan.\footnote{An initial deployment of the scheme, conducted under Harvard University's IRB, is underway in Uganda with 100 agricultural borrowers, thanks to a partnership with Makere University and generous support from the {\em Global Challenges in Economics and Computation}.}
We handle lenders with a minimum profit threshold on the belief that a borrower will repay a loan and handle lenders with and without liquidity constraints that limit how many loans they can make. 

The goal is to support a lender who wants to make loans to the best borrowers, where this is defined according to the belief aggregation rule of a lender and its liquidity constraints and profit thresholds. We achieve this through incentive alignment so that recommenders will have strict incentives to prefer to report their true beliefs over other possible reports. This is the concept of {\em strict properness} from the scoring rule literature. Here, we seek strict properness in the {\em interim}, so that it holds for any information of a recommender (any belief about the likelihood of repayment of borrowers) and in expectation with respect to a  prior on the beliefs of other recommenders about the likelihood of repayment of borrowers.
%
%

The  main results  are the following: 
\begin{enumerate}
    \item  A mechanism, the truncated Winkler mechanism, that is strictly proper for a reasonable ({\em grain of no veto}) technical condition for a lender without liquidity constraints and multiple recommenders and multiple borrowers, and for a monotone non-decreasing belief aggregation rule. We show that the mechanism is not incentive compatible for a liquidity-constrained borrower.
\item     A mechanism, the  Vickrey-Clarke-Groves (VCG) scoring mechanism, that is strictly proper under reasonable conditions (reasonably uniform aggregation weights across recommenders),  for a lender with or without liquidity constraints and multiple recommenders and multiple borrowers, and for a weighted linear belief aggregation rule. 
\item The VCG scoring mechanism also aligns incentives with recommenders wanting to receive larger weights in the aggregator, and thus  higher quality predictions when this leads to higher weights over time. Moreover, the VCG scoring mechanism can be configured to ensure that all recommenders have non-negative utility from participation, whatever the outcome from making loans.
\end{enumerate}

In regard to the first result, which makes use of the Winkler scoring rule, it is important that we use asymmetric scoring rules such that the minimum expected score  from the scoring rule is associated with the lender's threshold on minimum probability of repayment at which making a loan is profitable.  We develop the {\em grain of no veto} condition, which provides strict incentives, by reasoning about the interim utility and uncertainty faced by any given recommender.

While making use of the well known VCG mechanism, the application of the VCG scoring mechanism to this context is novel and  nonstandard. In particular, we use payments both in the normal sense of VCG {\em and also to construct the valuation functions of recommenders for different outcomes}. We make use  of outcome-contingent  payments to, in effect, give a recommender a  valuation functions for making a loan to a borrower that is proportional to the belief the recommender has as to the repayment probability of the borrower. By folding the typical VCG style payments on top,  we in effect take a constant (trivially proper but not strictly proper) scoring rule and generate an elicitation mechanism that is strictly proper. Moreover, the valuation functions are defined so that the  allocation rule of the mechanism corresponds to a belief aggregation model, and can embed  weights assigned by a belief aggregator to recommenders in an incentive-compatible way. This use of linear weighted belief aggregators corresponds naturally to well-studied belief aggregation systems \citep{SouleDavidAHfC}.

While our work is inspired by lending, these two mechanisms are broadly applicable in decision settings where interim uncertainty creates full support on the decision space from the agent perspective, and where outcomes are observed. Examples include employee screening, tenant screening, insurance underwriting, and service provider ratings.

The remainder of this paper is structured as follows. Section \ref{related_work} contains a brief literature review of prediction markets, decision markets, scoring rules and report aggregation. Section \ref{preliminaries_section} describes our recommendation gathering system and defines notation and key terms. Section \ref{winkler_section} discusses the unconstrained liquidity setting and the truncated Winkler mechanism. Section \ref{vcg_section} introduces the VCG scoring mechanism for the liquidity-constrained setting and gives an analysis of the incentive properties of the mechanism. 
Section \ref{conclusion_section} concludes and points to directions for future work.

\subsection{Related Work} \label{related_work}


One way to formulate the problem of gathering community lending recommendations \dcpadd{is as a {\em peer prediction} problem, i.e., as a problem of information elicitation without verification. The approach in peer prediction is to leverage correlation and mutual-information structure between reports to promote incentive alignment around true reports. A number of peer prediction mechanisms have been proposed, each requiring varying levels of knowledge on the part of the designer, on the kinds of reports, and on the task e.g. \citep{MillerNolan2005EIFT,WitkowskiP12,WitkowskiP12a,ShnayderAFP16,AgarwalMP017,JurcaF09,KongS19,RadanovicFJ16,Wang19,WaggonerC14}.  Indeed, this was the approach taken to belief elicitation for microfinance in~\citet{Hussam18}.

Our framing of the microfinance problem is that of information elicitation with verification, where a lender will observe whether or not a borrower makes a repayment or defaults in the future,} making this setting  well-suited for the methods of scoring rules, prediction markets, and decision markets.
{\em Scoring rules} are methods to elicit beliefs about uncertain future events where the outcome will be later observed~\citep{GneitingTilmann2007SPSR}. The basic framework is that a participant  reports her belief, $\hat{p}$ of the probability of an event. Her true belief is denoted as $p$. After the event is observed, participants are paid a reward (potentially negative) of $s(\hat{p},o)$, \dcpadd{where $o \in \mathcal{O}$ is the outcome of the event, $\mathcal{O}$ is a finite, exhaustive set of mutually-exclusive outcomes, and $s$ is a scoring rule mapping the report $\hat{p}$ and outcome $o$ to reward. A {\em proper scoring rule} is one in which the expected score from a truthful report is at least as great as the expected score from any non-truthful report, i.e.,

\begin{equation}
    \mathbb{E}_{o\sim p}[s(p,o)] \geq \EX_{o\sim p}[s(\hat{p},o)] \; ;\;\forall \hat{p} \neq p
\end{equation}

A {\em strictly proper scoring rule} replaces this inequality with a strict inequality.}
Common strictly proper scoring rules for the binary outcome case include the logarithmic scoring rule and the Brier or quadratic scoring rule. These are symmetric in the sense that the expected score when truthfully reporting is minimized at $p=0.5$ and  symmetric about that point. A useful modification is the class of {\em Winkler scoring rules} \citep{WinklerRobertL1994EPAS} which allow designers to set the minimum score point at any arbitrary location $c \in (0,1)$.

\dcpadd{{\em Prediction markets} can be used in a way that combines scoring rules with sequential elicitation from multiple participants, with the current market price reflecting the aggregate belief of the population about the outcome of an uncertain event~\citep{ChenP10}.  The market starts with an initial prediction $p_0 \in \Delta(\mathcal{O})$, where $\Delta(\mathcal{O})$ is the set of distributions on outcomes. In a prediction market with an automated market maker, for example the logarithmic market-scoring rule, agent $i$'s report $\hat{p}_i$ is in effect scored relative to the preceding report $\hat{p}_{i-1}$~\citep{hanson07}.}

Scoring rules and prediction markets are analyzed in a setting in which agents have no influence on whether information about whether an event is realized. This is not the case in the present paper, where agents' reports will determine who gets a loan. This brings us close to the 
{\em decision market } framework where a principal makes a decision  based on market prices~\citep{ChenYilingDMwG}. This creates new incentive challenges. For example, suppose there is  one loan to allocate and two  borrowers, with current market price 0.8 and 0.6 for each borrower, and an agent with corresponding beliefs  0.9 and 0.89. In a market-scoring rule context, the agent would have a higher expected payoff than truthful reporting  by leaving the price on borrower 1 unchanged and buying borrower 2 to a price of 0.89, so that the second borrower gets the loan.

\dcpadd{\citet{ChenYilingDMwG} provide a characterization of strict properness that requires randomization over decisions and full support on the set of possible decisions.  This can present a challenge to many applications, for example to decisions about  construction projects or the present context of loan decisions, and a hurdle to the real-world implementation of decision markets. As we discuss in Sections \ref{winkler_section} and \ref{vcg_section}, we provide a counterpoint to this requirement of  randomization with full support. 
The key difference is that the incentive analysis is conducted in the {\em interim} when an agent knows its own belief report but is uncertain about the belief reports of others, this uncertainty  maintained through  simultaneous reports. This interim uncertainty enables strict incentive alignment for a deterministic decision rule that lends to the set of borrowers that are most likely not to default.}

 VCG concepts have been used together with scoring rules but not in the way described here. \citet{PapakonstantinouAthanasios2011Mdft}  use a two-stage elicitation mechanism whereby agents participate in a second price procurement auction for information with the winner paid using a scoring rule based on the cost of the second-lowest-cost agent.  
%
Other work has considered  the problem of selection of  agents from within the same group as those who  provide yes/no approval information on others~\citep{AlonNoga2011Sous}. The goal is to select agents with the maximum number of approvals from others and there is no downstream observation of an uncertain event. Whereas we assume the recommenders is disjoint from the set of potential borrowers, \citet{AlonNoga2011Sous} focus on the incentive issues that arise when 
this is not the case. 

This work also relates to belief aggregation, which we briefly review in the Appendix.

\section{Preliminaries} \label{preliminaries_section}

A lender has a set of candidate borrowers $M=\{1,\ldots,m\}$ and recruits a set $N=\{1,\ldots,n\}$ of recommenders who know the candidate borrowers personally and provide reports on the $m$ borrowers. Each recommender $i\in N$ has a subjective belief $p_{iq} \in [0,1]$ of the likelihood with which candidate borrower $q$ will repay a loan. 
We write $p_i=(p_{i1},\ldots,p_{im})$, $p_q=(p_{1q},\ldots,p_{nq})$, and $p=(p_1,\ldots,p_n)$.
We also refer to  belief $p_i$ as the {\em type} of the recommender. 
We let ${\mathcal D}$ denote a {\em  prior} on beliefs, such that $p\sim {\mathcal D}$. We write $p_{-i}=(p_1,\ldots,p_{i-1},p_{i+1},\ldots,p_n)$, and write $p_{-i}\sim {\mathcal D}_{-i}$, marginalizing out over recommender $i$. We assume ${\mathcal D}$ and ${\mathcal D}_{-i}$ are common knowledge.

Recommender $i$ makes a {\em report $\hat{p}_{iq} \in [0,1]$} to the lender (principal), and we allow  $\hat{p}_{iq}\neq p_{iq}$. We use  $\hat{p}_i$ to denote the profile of all belief reports of recommender $i$. In the mechanisms that we design, recommenders make reports independently with no knowledge of other recommenders' reports. 
The {\em  repayment outcome} for a borrower $q$ who receives a loan is a binary variable,  $o_q \in \{0,1\}$, with $1$ representing  repayment and $0$ representing default.

The lender makes a decision about which borrowers will receive a loan. We assume that the lender has a {\em profit threshold} $c \in [0,1]$, such that the lender makes profit when making a loan where the repayment probability is $c$ or higher. 
The lender forms a {\em belief} about the repayment probability of a borrower $q\in M$ with an aggregation function $B_q(p_q)$, which represents the lender's belief where $p_q=(p_{1q},\ldots,p_{nq})$.
We assume this is weakly monotone increasing in $p_i$, for each $i$.
We also sometimes work with  a linear aggregator,  $B_q(p_q)=\sum_{i\in N}w_i p_{iq}$,  with {\em weight} $w_i>0$ on  recommender $i$ and $\sum_{i\in N}w_i=1$.
The lender also has a {\em liquidity constraint} $K\leq  m$, and if $K<m$ then can only make loans to a limited number of  borrowers.

\begin{definition}[Elicitation mechanism]
We design an {\em elicitation mechanism} $\mathcal{M}=(x,t,s)$:
\begin{enumerate}
    \item Elicit belief reports $\hat{p}=(\hat{p}_1,\ldots,\hat{p}_n)$ from  recommenders
    \item Determine the set of borrowers, $x(\hat{p})\in \{0,1\}^m$, that will receive a loan, such that $\sum_{q\in M}x_q(\hat{p})\leq K$, and define  two-part payments:
    \begin{enumerate}
        \item  An {\em immediate payment  } $t_i(\hat{p})\in \mathbb{R}$ {\bf made by each} recommender $i\in N$
        \item An {\em outcome-contingent payment} $s_{iq}(\hat{p}_{iq},o_q)\in \mathbb{R}$ {\bf made to each} recommender $i$ for each borrower $q\in a=x(\hat{p})$, i.e., for each borrower for which $a_q=1$.
    \end{enumerate}
\end{enumerate}
\end{definition}
\medskip

Given reports $\hat{p}$ and outcome profile $o=(o_1,\ldots,o_m)$, the {\em realized utility} to recommender $i$ is
\begin{align}
    u_i(\hat{p}_i,\hat{p}_{-i},o) = \sum_{q\in x(\hat{p})}s_{iq}(\hat{p}_{iq},o_q) - t_i(\hat{p}). 
\end{align}

Here, $\hat{p}_{-i}=(\hat{p}_1,\ldots,\hat{p}_{i-1},\hat{p}_{i+1},\ldots,\hat{p}_n)$.
%
%
The {\em utility} for recommender $i$ with belief $p_i$ is
\begin{align}
U_i(p_i,\hat{p}_i,\hat{p}_{-i})&=\sum_{q\in x(\hat{p})} \left( p_{iq}s_{iq}(\hat{p}_{iq},1) + (1-p_{iq})s_{iq}(\hat{p}_{iq},0)\right) - t_i(\hat{p})\\
&= \sum_{q\in x(\hat{p})} \mathbb{E}_{o_q\sim p_{iq}}[s_{iq}(\hat{p}_{iq},o_q)] - t_i(\hat{p}) \notag
\end{align}


This quantity is  ex post with respect to the reports of others, and takes an expectation over borrower outcomes  with respect to the beliefs of recommender $i$.
 It is useful to interpret the outcome-contingent payment to the recommender as inducing a term that plays a similar role as an agent's valuation in mechanism design, where $\mathbb{E}_{o_q\sim p_{iq}}[s_{iq}(\hat{p}_{iq},o_q)]$ is the recommender's ``value" for the lender's decision to lend to borrower $q$. 

\medskip

 There are a number of possible desiderata for an elicitation mechanism in this setting. 
\begin{itemize}
    \item  {\em Allocative efficiency} means that the mechanism allocates to the borrowers with the maximum probability of repayment amongst those better than the profit threshold $c$. For lender belief $B_q(\hat{p}_q)$, this requires that 
\begin{align} 
    x(p) \in  & \argmax_{a\in \{0,1\}^m}  \sum_{q \in M: B_q(p_q)>c} B_q(p_q) \times a_q \\
\mbox{s.t.} \quad &    \sum_{q\in M} a_q \leq K \notag
\end{align}
 

 \item {\em Weak ex post incentive compatibility} (weak EPIC) means that each recommender's utility is weakly maximized by reporting truthfully, regardless of the reports of other recommenders  (or \textit{ex post proper} in the language of the scoring rule literature, adapted here to also consider the effect of reports of others). This is 
    \begin{align}
       U_i(p_i,p_i,\hat{p}_{-i}) \geq U_i(p_i,\hat{p}_i,\hat{p}_{-i})\; ; \forall i, \ \forall p_i, \ \forall \hat{p}_i, \ \forall \hat{p}_{-i}
    \end{align}

 \item    {\em Strict ex post incentive compatibility} (strict EPIC) means that each recommender's utility is strictly maximized by reporting truthfully, regardless of the reports of other recommenders  (or  \textit{ex post strict proper} in the language of the scoring rule literature, adapted here to also consider the effect of reports of others). This is 
    \begin{align}
     U_i(p_i,p_i,\hat{p}_{-i}) > U_i(p_i,\hat{p}_i,\hat{p}_{-i})\; ; \forall i, \ \forall p_i, \ \forall \hat{p}_i \neq p_i, \ \forall \hat{p}_{-i}
    \end{align}
    
These incentive concepts are  ex post with respect to the reports of others and interim with respect to   a recommender's own belief on the outcome of making loans to borrowers, and defined before repayment outcomes are observed.

  \item {\em Strict interim incentive compatibility} (strict IIC) means that each recommender's interim utility, considering beliefs of others, is strictly maximized by reporting truthfully (or \textit{strict properness} in the language of the scoring rule literature). This is
   \begin{align}
        \mathbb{E}_{p_{-i}\sim {\mathcal D}_{-i}}[U_i(p_i,p_i,p_{-i})] > \mathbb{E}_{p_{-i}\sim {\mathcal D}_{-i}}[U_i(p_i,\hat{p}_i,p_{-i})]\; ; \forall i, \ \forall p_i, \ \forall \hat{p}_i \neq p_i
    \end{align}

\item  \textit{Ex post Individually Rational} (IR) means that all recommenders that make a truthful report have a non-negative expected utility once loans are allocated, but before repayment outcomes are observed.  This is
 \begin{align}
      U_i(p_i,p_i,\hat{p}_{-i})\geq 0\; ; \forall i, \ \forall p_i, \ \forall \hat{p}_{-i}. 
 \end{align}


 \item \textit{Strong ex post IR} means that all recommenders that make a truthful report have a non-negative realized utility even after repayment outcomes are observed. This is 
\begin{align}
    u_i(p_i, \hat{p}_{-i},  o) \geq 0\; ; \forall i, \ \forall p_i, \ \forall \hat{p}_{-i}, \ \forall o. 
    \end{align}
    
   The first concept of IR is  ex post with respect to the  reports of other recommenders. For this reason, we adopt the phrasing strongly ex post IR for the second notion, which holds once outcomes are observed.
    
    \end{itemize}

\section{Unconstrained-Liquidity Setting: Truncated Winkler Mechanism} 
\label{winkler_section}
In this section, we use strictly proper scoring rules to  design  elicitation mechanisms for the unconstrained liquidity setting, where the problem  decomposes  to make a separate decision for each borrower. 
We make use of the {\em  Winkler scoring rule}.   Relative to the otherwise more general VCG scoring mechanism, this approach can accommodate non-linear belief aggregation rules.

\subsection{One Recommender, One  Borrower} 
\label{1_rec_1_borrower}

We first consider the  case of a single recommender and a single  borrower, and  drop subscripts for notational simplicity. 
\dcpadd{Without care, this setup can be prone to the  problem outlined in~\citet{ChenYilingDMwG} for decision markets, and a recommender may be incentivized to misreport to ensure that the borrower receives a loan and  they have a chance of being compensated. To illustrate this, consider setting the immediate payment $t$ to zero, and using making use of a 
{\em truncated quadratic scoring rule} for the outcome-contingent payment:
\begin{equation}
\begin{IEEEeqnarraybox}[][c]{rcl}
s^{TQ}(\hat{p}, o) = \begin{cases} 
      2 \hat{p} - \hat{p}^2 - (1-\hat{p})^2 & \mbox{if $\hat{p} > c, o = 1$, and} \\
      
      2 (1 - \hat{p}) - \hat{p}^2 - (1-\hat{p})^2 & \mbox{if $\hat{p} > c, o = 0$.}\\
   \end{cases}
\IEEEstrut\end{IEEEeqnarraybox}
\end{equation}

The rule is ``truncated" because it only defines a score when report $\hat{p}>c$. The recommender receives no outcome-contingent payment when  
$\hat{p} \leq c$ and a loan is not made.}
\begin{figure} 
    \begin{center}
    	\includegraphics[width=9 cm]{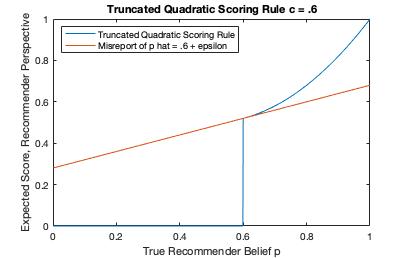} \\
        \caption{Expected utility for a truncated quadratic scoring rule under truthful reporting vs.~expected utility when reporting $\hat{p} = .6 + \epsilon$, for small $\epsilon>0$. The threshold is set to $c=0.6$.
        \label{quadratic_breakage_winkler}}
    \end{center}
\end{figure}

\dcpadd{Let $U(p)$  denote the {\em expected utility under truthful reporting}. For $p>c$ this is $\mathbb{E}_{o\sim p}[s^{TQ}(p,o)]$ and otherwise this is zero. 
See Figure \ref{quadratic_breakage_winkler}, for the case of $c=0.6$.  A recommender with true belief $p < 0.6$ can achieve a higher  utility 
by misreporting $\hat{p} = 0.6 + \epsilon$ for small $\epsilon>0$.
This is akin to the problem with incentive alignment in
decision markets.}

One idea to address this is to add an immediate payment $t(\hat{p})$,  set to compensate for the forfeited payment in the case that the report $\hat{p}\leq c$. For the quadratic scoring rule, we would set 
\begin{align} \label{quad_pmt_t}
    t^{TQ}(\hat{p})&=\left\{
    \begin{array}{ll}
   -[c^2 + (1-c)^2] & \mbox{if $\hat{p}\leq c$,}\\
     0 &\mbox{otherwise.}
    \end{array}
    \right.
\end{align}

We negate the quantity here, because the convention is that this is the payment made {\bf by} a recommender to the mechanism. 
In this way, there is continuity in the expected payment across this report threshold.
However, this elicitation mechanism aligns incentives as long as threshold $c\geq  0.5$, but not otherwise. 

Figure~\ref{proper_truncated_quad_fig} illustrates the combined effect of the one-time and outcome-contingent payments, plotting $U(p)$. 
For $c=0.6$, $U(p)$ is convex, which is a sufficient condition for the rule to be proper~\citep{GneitingTilmann2007SPSR}. In particular, it is strictly convex and strictly proper for $p>c$ and weakly proper for $p\leq c$. The strictness for beliefs above $c$ ensures the lender will make the correct decision for any optimal report of the recommender. However, for  $c < 0.5$, $U(p)$ is non-convex and a recommender with belief $p \in (c,1-c)$ can obtain a higher  utility by reporting $\hat{p} \leq c$. This is because the non-truncated quadratic scoring rule $s^Q$ is symmetric with $G^Q(p) = G^Q(1-p)$, for $G^Q(p)=\mathbb{E}_{o\sim p}[s^{Q}(p,o)]$, and with a minimum  $G^Q(p)$ at $p = .5$. 

\begin{figure} 
    \begin{center}
    	\includegraphics[width=9 cm]{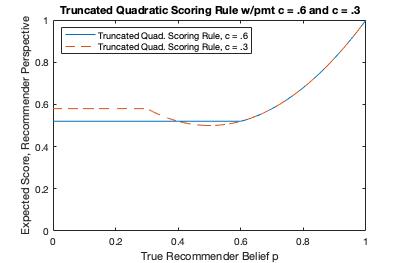} \\
        \caption{Expected utility for a truncated quadratic scoring rule together with immediate payment $t^{TQ}$  for threshold values $c = .6$ and $c = .3$. Note the properness when $c \geq 5$ and the lack of properness when $c < .5$ \label{proper_truncated_quad_fig}}
    \end{center}
\end{figure}

To avoid this, we  use an  asymmetric scoring rule whose minimum expected score under truthful reporting is at $p = c$. Let $s$ be any symmetric proper scoring rule, and consider some $c \in (0,1)$. The  {\em Winkler scoring rule}
$s^W$~\citep{WinklerRobertL1994EPAS}  is
\begin{equation}
    s^W(\hat{p}, o) = \frac{s(\hat{p}, o)-s(c, o)}{T(c,\hat{p})}, 
\quad \mbox{where}    \quad 
    T(c,\hat{p}) = \begin{cases} 
      s(0,0) - s(c,0) & \mbox{if $\hat{p} \leq c$} \\
      
      s(1,1) - s(c,1) & \mbox{otherwise.}  
   \end{cases}
   \end{equation}

The Winkler rule $s^W$ is (strictly) proper when $s$ is (strictly) proper. We choose to build $s^W$ from the logarithmic scoring rule, obtaining:
\begin{equation}
\begin{IEEEeqnarraybox}[][c]{rcl}
s^W(\hat{p}, o) & \ = \ & \begin{cases} 
      \frac{\ln(\hat{p})-\ln(c)}{ - \ln(c)} & \mbox{if $\hat{p} > c, o = 1$}\\
      \frac{\ln(1-\hat{p})-\ln(1-c)}{ - \ln(c)} & \mbox{if $\hat{p} > c, o = 0$.}
   \end{cases}
\IEEEstrut\end{IEEEeqnarraybox}
\label{eq:winkler12}
\end{equation}

\begin{figure}
    \begin{center}
    	\includegraphics[width=8 cm]{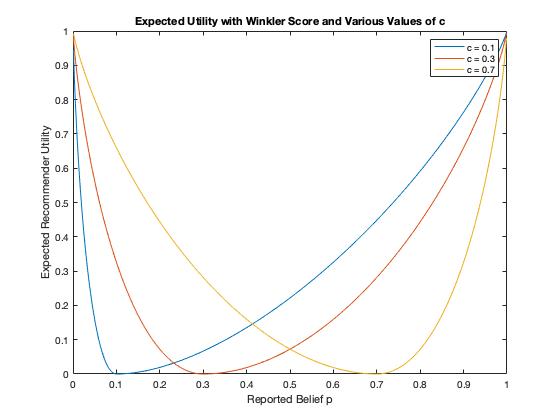} 
    	\includegraphics[width=8 cm]{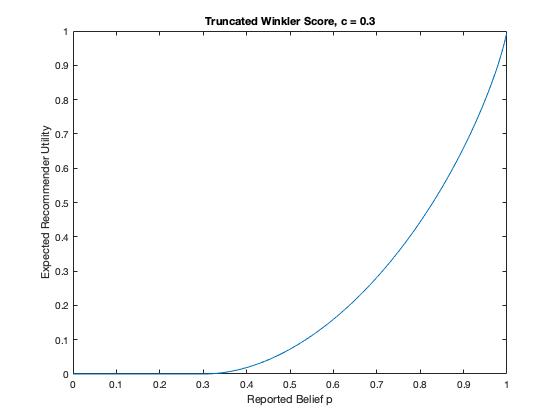} 
        \caption{{\bf Left}: Expected score $G^W(p)$ with truthful reporting under the Winkler Scoring Rule based on the Logarithmic Scoring Rule.  {\bf Right}: Expected utility $U(p)$ with truthful reporting for the truncated Winkler elicitation mechanism and lender threshold $c=0.3$. \label{truncated_winkler}}
    \end{center}
\end{figure}


Figure~\ref{truncated_winkler} (left) shows the   expected score  
$G^W(p)=\mathbb{E}_{o\sim p}[s^{W}(p,o)]$ from the Winkler rule. $G^W(p)=0$  for  $p=c$.
This leads to the following elicitation mechanism.
\begin{definition}[Truncated Winkler elicitation mechanism (1 recommender, 1 borrower)]
The {\em Truncated Winkler elicitation mechanism} for 1 recommender and 1 borrower  and lender profit threshold $c$ is defined as following:
\begin{itemize}
    \item Allocation: $x(\hat{p})=1$ if $\hat{p}>c$ and $x(\hat{p})=0$ otherwise
    \item Payment
    \begin{itemize}
        \item Immediate payment: zero
        \item Outcome-contingent payment: $s(\hat{p},o)=s^W(\hat{p},o)$ as per Eq.~\eqref{eq:winkler12}.
    \end{itemize}
\end{itemize}
\end{definition}

Figure~\ref{truncated_winkler} (right)  shows the  expected utility $U(p)$ from truthful reporting 
in a truncated Winkler mechanism  with $c=0.3$. The utility is convex and given by
\begin{equation} \label{single_winkler_mech}
\begin{IEEEeqnarraybox}[][c]{rcl}
U(p) & \ = \ & \begin{cases} 
      0 & \mbox{, if $p \leq c$} \\
      p\frac{\ln(p)-\ln(c)}{ - \ln(c)} + (1-p)\frac{\ln(1-p)-\ln(1-c)}{ - \ln(c)} & \mbox{otherwise.} 
   \end{cases}
\IEEEstrut\end{IEEEeqnarraybox}
\end{equation}

 %
\begin{theorem}
For the one-borrower, one-recommender case, the truncated Winkler elicitation mechanism with lender profit threshold $c$
is ex post IR, and strict   proper
for beliefs $p > c$ and  proper for all beliefs.
\label{thm:1b1rwinkler}
\end{theorem}
\begin{proof}
For a single recommender setting, strict ex post proper and ex post proper are equivalent to simply strict proper and proper, since there are no reports of others.
The expected score $G^W(p)$ from the Winkler rule under truthful reporting is strictly positive when $p>c$~\citep{WinklerRobertL1994EPAS}.  Since reports of $\hat{p} \leq c$ yield payment of $0$, risk-neutral agents will always prefer  $\hat{p} > c$ when $p > c$. Given that $s^W$ is strictly proper, the full mechanism is also strictly proper when $p > c$. When $p \leq c$, the convexity of $U(p)$ and  the fact that $U(c) = 0$ guarantees that 
%
a recommender weakly maximizes its expected payment from truthful reporting when  $p \leq c$. This gives properness for all beliefs. 
Individual rationality is immediate, since $U(p)\geq 0$.
$\blacksquare$
\end{proof}



\subsection{Multiple Recommmenders, Multiple  Borrowers}

We now consider the multiple recommenders, multiple borrowers case. Because liquidity is unconstrained, the lending decision decomposes easily across borrowers. 

A new aspect of the analysis is that with multiple recommenders we can  obtain strict interim IC (strict properness) by reasoning about the interim uncertainty a recommender has about the reports of other recommenders. We bring back subscripts, and  recommender $i$'s report on borrower $q$ is denoted  $\hat{p}_{iq}$ and borrower $q$'s repayment outcome is $o_q$. 
%


For borrower $q$, a  loan is made if and only if $B_q(\hat{p}_q)>c$.
In defining the outcome-contingent payment, we let $c_{iq}$ denote the {\em minimum report by recommender $i$ on borrower $q$  such a loan is made to $q$}, i.e.,
\begin{equation}
    c_{iq} \ \triangleq \  \underset{p' \in [0,1]}{\mathrm{inf}}p'\; \mathrm{ s.t.} \; B_q(\hat{p}_{iq}=p',\hat{p}_{-iq}) > c,
    \label{eq:ciq}
\end{equation}
where $p_{-iq}=(p_{1q},\ldots,p_{i-1,q},p_{i+1,q},\ldots,p_{nq})$. $c_{iq}$ is a function of the reports of others but we leave this implicit. By weak monotonicity of the belief aggregator, for any $\hat{p}_{iq}>c_{iq}$ we have that borrower $q$ receives a loan.
For example with a linear aggregator  $B_q(\hat{p}_q) = \sum_{i \in N} w_i \hat{p}_{iq}$, we have 
\begin{equation}
\label{ciqdfn}
    c_{iq} = \mathrm{min} (1, \mathrm{max} (0, \frac{1}{w_i}(c - \sum_{j \neq i} w_j \hat{p}_{jq}) ) ).
\end{equation}



Each recommender's outcome-contingent payment  follows the  Winkler scoring rule~\eqref{eq:winkler12}, with Winkler parameter $c$ replaced by $c_{iq}$. We let $s^W_{iq}(\cdot,\cdot)$ denote this modified Winkler scoring rule.
%

\begin{definition}[Truncated Winkler elicitation mechanism ($n$ recommenders, $m$ borrowers)]
The {\em Truncated Winkler elicitation mechanism} with unconstrained lender liquidity, lender profit threshold $c$, and monotone belief aggregation $B_q$, is defined as following:
\begin{itemize}
    \item Allocation: for each borrower $q$, $x_q(\hat{p})=1$ if $B_q(\hat{p}_q)>c$ and $x_q(\hat{p})=0$ otherwise
    \item Payment:
    \begin{itemize}
        \item Immediate payment: zero
        \item Outcome-contingent payment: for each borrower that receives a loan, $s_{iq}(\hat{p}_{iq},o_q)=s^W_{iq}(\hat{p}_{iq},o_q)$ as per Eq.~\eqref{eq:winkler12}, and with parameter $c$ in the Winkler rule set to $c_{iq}$ as per Eq.~\eqref{eq:ciq}.
    \end{itemize}
\end{itemize}
\end{definition}

\medskip

Lemma~\ref{winkler_properness_lemma} follows immediately from the proof of Theorem~\ref{thm:1b1rwinkler}, recognizing that the thresholds $c_{iq}$ play the role of $c$ (and allowing for the incentive and IR properties to be stated ex post, i.e., for any reports of others).
\begin{lemma} \label{winkler_properness_lemma}
   \sloppy For unconstrained liquidity, multiple borrowers and multiple recommenders, and a weak monotone-increasing aggregation function $B_q$, the truncated Winkler mechanism 
   is  ex post individually rational, and ex post proper for a report of recommender $i$ on borrower $q$ when $p_{iq} \leq c_{iq}$ and strictly ex post proper when $p_{iq} > c_{iq}$.
\end{lemma}

We now state our first main theorem. For this, define {\em no veto} for $i$ and $q$ at threshold $c$ to mean $B_q(0,p_{-i,q})>c$, i.e., beliefs of others such that even with $p_{iq}=0$ from recommender $i$   the borrower will receive a loan under truthful reports,  and thus also with any  belief of recommender $i$ (by the monotonicity of $B_q$). 
\begin{definition}[Grain of no veto]
The   distribution on beliefs satisfies {\em a grain of no veto at  $c$} when 
$\mathbb{P}[p_{-i,q} \sim {\mathcal D}_{-i} : B_q(0,p_{-i,q}) > c] > 0$, for all recommenders $i$, all borrowers $q$, i.e., 
no veto is satisfied with non-zero measure of the type distribution.
\end{definition}

\begin{theorem} \label{B_q_strict_proper}
For unconstrained liquidity, multiple borrowers and multiple recommenders, and a weak monotone-increasing aggregation function $B_q$, the truncated Winkler mechanism is ex post IR, and also strictly interim IC  (strict proper)  when the distribution on beliefs satisfies grain of no veto.
\end{theorem}
\begin{proof} 
Per Lemma \ref{winkler_properness_lemma}, when $p_{iq} > c_{iq}$, the Winkler scoring mechanism is strictly ex post proper with regard to $i$'s report on $q$. Otherwise, it is  ex post proper. By the grain of no veto property, for  any belief of $i$ on $q$ $p_{iq}$, there is non-zero probability that $p_{iq}>c_{iq}$. This implies    $\mathbb{E}_{p_{-i}\sim {\mathcal D}_{-i}}[U_i(p_i,p_i,p_{-i})] > \mathbb{E}_{p_{-i}\sim {\mathcal D}_{-i}}[U_i(p_i,\hat{p}_i,p_{-i})]$, for all $\hat{p}_{-i}\neq p_i$.
IR follows immediately from Lemma~\ref{winkler_properness_lemma}.
$\blacksquare$
\end{proof}


For a weighted linear aggregator, the grain of no veto condition requires (for all $i$, all $q$) that $\mathbb{P}[p_{-i,q} \sim {\mathcal D}_{-i} : \sum_{j\neq i} w_j p_{jq} > c] > 0$. That is, there is some non-zero probability that the weighted sum over reports of all but one recommender is large enough. 
As a special case, we can state a  corollary for the case that  beliefs have full support on $[0,1]$.
\begin{corollary}
For  unconstrained liquidity, multiple borrowers and multiple recommenders, and a belief distribution with full support, the truncated Winkler mechanism  
with the weighted linear aggregator
is ex post IR, and also strictly interim IC  (strict proper) for lender profit threshold $c$ when $\max_i w_i<1-c$, which is only possible when $c<(n-1)/n$.
\end{corollary}
\begin{proof}
 Since $\max_i w_i<1-c$, then $\sum_{j\neq i} w_j>c$, for all $j\in N$. From this, we have $\mathbb{P}[p_{-i,q} \sim {\mathcal D}_{-i} : \sum_{j\neq i} w_j p_{jq} > c] > 0$, by full support, and thus the grain of no veto condition. Moreover, since $\max_i w_i \geq 1/n$, we need $1/n<1-c$ and thus $c<(n-1)/n$.
\end{proof}

As the threshold increases the system needs more recommenders to provide strict properness and the lender becomes less able to put a very large weight on any single recommender.

\section{Constrained-Liquidity Setting: The VCG Scoring Mechanism} 
\label{vcg_section}

In this section, we give a mechanism that provides strict properness for the possibly constrained-liquidity case, i.e.,  handling  $K<m$ in addition to $K=m$. 
%
We first illustrate the failure of the truncated Winkler mechanism, and then introduce the VCG scoring mechanism that achieves strict properness together with a linear belief aggregator.

\subsection{Failure of the Truncated Winkler Mechanism}

The following theorem also implies that the Truncated Winkler mechanism  fails to be strict interim IC  under the grain of no veto property. 
\begin{theorem} \label{winkler_fail_theorem}
    In the constrained-liquidity setting with more than one recommender, the Truncated Winkler mechanism is not weakly ex post IC (ex post proper). 
\end{theorem}

\begin{proof} 
Consider three recommenders, two  borrowers, budget $K = 1$, lending threshold $c = .5$, and $s^W$ based on the logarithmic scoring rule. Suppose an aggregator that is a simple average of reports.  The beliefs are as in Table~\ref{tab:winklerfail}.  If all recommenders report truthfully, lender's belief will be $~.57$ and $.55$, for borrowers 1 and 2 respectively, and borrower 1 will be allocated. The expected utility (payment) of recommender 2  will be $p_{2,1} s^W(p_{2,1}, 1) + (1-p_{2,1}) s^W(p_{2,1}, 0) = .4 \frac{ln(.4) - ln(.2)}{-ln(.2)} + (1-.4) \frac{ln(1-.4) - ln(1-.2)}{-ln(.2)} = .4*.43 + .6*(-.18) = .07$. If, recommender 2 misreports $\hat{p}_{21}=0$, then the lender beliefs will be $.43$ and $.55$ for borrowers 1 and 2, and  borrower 2 will be allocated. In this case, recommender 2's expected utility will be $p_{2,2} s^W(p_{2,2}, 1) + (1-p_{2,2}) s^W(p_{2,2}, 0) = .85 \frac{ln(.85) - ln(.7)}{-ln(.7)} + (1-.85) \frac{ln(1-.85) - ln(1-.7)}{-ln(.7)} = .85*.54 + .15*(-1.94) = .17$. Similar examples can be constructed for any number of recommenders greater than one. $\blacksquare$
\end{proof}

\begin{table}
\begin{center}
     \begin{tabular}{|| c | c | c | c |} 
         \hline
          & Recomender 1 & Recommender 2 & Recommender 3 \\ [0.5ex] 
         \hline\hline
         Belief on Borrower 1 & .7 & .4 & .6 \\ 
         \hline
         Belief on Borrower 2  & .4 & .85 & .4 \\
         \hline
        Borrower 1 threshold $c_{i1}$ & .5 & .2 & .4 \\
         \hline
      Borrower 2 threshold   $c_{i2}$ & .25 & .7 & .25 \\
         \hline
         Expected utility, Honest & .12 & .07 & .09 \\
         \hline
         Expected utility, Recommender 2 Misreport & .04 & .17 & .04 \\ [1ex] 
         \hline
    \end{tabular}
    
\end{center}
\caption{ \label{tab:winklerfail} Perverse incentives with the Truncated Winkler mechanism in the liquidity-constrained setting.}
\end{table}

\subsection{The VCG Scoring Mechanism}

%

For the VCG scoring mechanism, we define the outcome-contingent payment to be that of a {\em constant scoring rule}, and
\begin{equation}
    s^{VCG}_{iq}(\hat{p}_{iq},o_q) = \begin{cases} 
      w_i, \mbox{if $o_q = 1$} \\
      0, \mbox{otherwise.}
   \end{cases}
\end{equation}

This scoring rule is trivially proper, but not strictly proper (it doesn't depend on the report). Interestingly, this will  provide strictly proper incentives when embedded within the framework of the VCG mechanism. 
%
It is useful to define the {\em value function} of a recommender for loan decisions $a\in \{0,1\}^m$ in the context of the VCG scoring mechanism as
\begin{align}
    v_i(a) & \triangleq \sum_{q \in M} a_q\times \EX_{o_q \sim p_{iq}} [s^{VCG}_{iq}(p_{iq},o_q)]=\sum_{q\in M}a_q w_ip_{iq}.
\end{align}
%
Similarly, we define the {\em reported value function} in the context of the VCG scoring mechanism as 
\begin{align}
    \hat{v}_i(a) & \triangleq \sum_{q \in M} a_q\times \EX_{o_q \sim \hat{p}_{iq}} [s^{VCG}_{iq}(\hat{p}_{iq},o_q)]=\sum_{q\in M}a_q w_i\hat{p}_{iq}.
\end{align}

These play the typical role of valuations and reported valuations in the analysis of the incentive properties of a VCG mechanism. 
A non-standard aspect  is that the weight $w_i$ is  under control of the mechanism designer because it reflects the payment made  in the event that a borrower makes a repayment (the outcome-contingent payments).
%
As the designer has direct control over this aspect of the value function, we can use these weights without using a weighted VCG mechanism. 

%

This  aligns the  allocation of loans that maximizes the total reported value with the allocation that lends to the borrowers who are best in the sense of having maximum total weighted  reported beliefs from recommenders, and thus are most likely to repay according to the lender's linear weighted belief aggregator.  That is, the value-maximizing allocation is just:
%
\begin{align} 
   \max_{a\in \{0,1\}^m}  & \sum_{q \in M} a_q \left(\sum_{i \in N}  w_i \hat{p}_{iq}\right)\\
    \mbox{s.t.} \quad &    \sum_{q\in M} a_q \leq K. \notag
\end{align}
%
%

To introduce a lending threshold $c>0$, we can also add $K$ imaginary {\em reserve borrowers} to the system and a {\em reserve recommender} with weight 1 who reports $c$ for a loan decision to each of these borrowers and 0 for other borrowers. 
%
%
We leave the weights to other recommenders unchanged. Adopting $R$ to represent the set of reserve borrowers, the modified allocation rule is
\begin{align} \label{eq:15}
    x^{VCG}(\hat{p}) &\in  \argmax_{a\in \{0,1\}^m}  \left[  \sum_{q \in M} a_q \left( \sum_{i \in N}  w_i \hat{p}_{iq}\right) + \sum_{q\in R}a_q c\right]\\
    \mbox{s.t.} \quad &    \sum_{q\in M\cup R} a_q \leq K. \notag
\end{align}

Going forward we will  incorporate  the reserve recommender into the set of $N$ recommenders 
and the $K$ reserve borrowers into the set $M$ of borrowers.
\begin{definition}[VCG scoring mechanism ($n$ recommenders, $m$ borrowers)]
The {\em VCG scoring mechanism} with possibly constrained lender liquidity, lender profit threshold $c$, and linear weighted belief aggregation $B_q$ with weights $w=(w_1,\ldots,w_n)$,  is defined as following:
\begin{itemize}
    \item Allocation: adopt $x^{VCG}(\hat{p})$
    \item Payment (no payments are made by or collected from the reserve recommender):
    \begin{itemize}
        \item Immediate payment: %
\begin{equation}
    t^{VCG}_i(\hat{p}_i, \hat{p}_{-i}) = \sum_{j \neq i} \hat{v}_j(x^{-i}(\hat{p}_{-i})) - \sum_{j \neq i} \hat{v}_j(x(\hat{p}_i,\hat{p}_{-i})),
\end{equation}
where $x^{-i}$ is the  allocation decision that would be made without $i$ present, i.e., ignoring the reports from recommender $i$. 
        \item Outcome-contingent payment: for each borrower that receives a loan, $s^{VCG}_{iq}(\hat{p}_{iq},o_q)=w_i$ if $o_q=1$ and $0$ otherwise.
    \end{itemize}
\end{itemize}
\end{definition}


%

The realized utility  of recommender $i$ after repayment outcomes are known is
\begin{equation}
    u_i(\hat{p}_i, \hat{p}_{-i}, o) \triangleq  \sum_{q \in x^{VCG}(\hat{p})} w_i o_q - t^{VCG}_i(\hat{p}_i, \hat{p}_{-i}).
\end{equation}

%


\subsection{Strict Properness of the VCG Scoring Mechanism}


\begin{theorem} \label{vcg_dsic_1}
    The VCG Scoring Mechanism is efficient,  satisfies weak ex post IC (ex post proper), and is ex post individually rational.
\end{theorem}
 
 Once valuation functions are set-up to correspond to aggregate belief reports this proof follows the standard recipe for the IC and IR properties of a VCG mechanism ; see the Appendix.
 

We also want strict interim IC (strict properness), so that it is a unique best response of a recommender to report its true beliefs.

We define an \textit{equal-shift misreport} as a misreport $\hat{p}_i\neq p_i$ for which $p_{iq}-p_{iq'}=\hat{p}_{iq}-\hat{p}_{iq'}$ for every $q$, every $q'$. 
We say a mechanism is {\em strictly proper up to equal-shift misreports} if truthful reporting is a unique best response, maximizing interim utility except for possible tie-breaking amongst  equal-shift misreports.
%
%
%


 We now state the first of our two main theorems.
 %

\begin{theorem} \label{equal_shift_thm} \label{thm:1d}
For constrained liquidity $(K<m)$, two or more borrowers, three or more recommenders, and a belief distribution with full support, the VCG Scoring mechanism without a reserve (i.e., $c=0$) is strict interim IC (strictly proper) up to equal-shift misreports when $\max_i [w_i] < 1/2$.
\end{theorem} 
\begin{proof}
Three or more recommenders are required for $\max_{i'} [w_{i'}]<1/2$. Two or more borrowers allows for constrained liquidity.
We consider recommender $i$, belief $p_i$, any $\hat{p}_i\neq p_i$ that is not an equal-shift misreport, and establish a non-zero measure on the  beliefs $p_{-i}$ of others such that the allocation changes in a way that reduces the total value (i.e., not selecting the borrowers with the top $K$ aggregate belief of repayment).
 Since VCG is weakly ex post IC (Theorem~\ref{vcg_dsic_1}), this establishes strict interim IC  up to equal-shift misreports.
 
For any $q$, let $B(q)$ denote the aggregate belief on $q$ at  $p_i$ and $\hat{B}(q)$ at report $\hat{p}_i$. If $p_{-i}$ satisfies $p_{jq} = \frac{1/2 - w_i p_{iq}}{\sum_{j' \neq i} w_{j'}} =p^*_{q}$, $\forall j \neq i$, then $B(q) = 1/2$.
This belief $p_{-i}$ is feasible by full support, and since for $p_{iq}=0$ we have $p^*_{q}=(1/2)/\sum_{j'\neq i}w_{j'}<1$  since $\sum_{j'\neq i}>1/2$ from $w_i<1/2$. For $p_{iq}=1$ we have $p^*_{q}=(1/2-w_i)/\sum_{j'\neq i}w_{j'}>0$ since $w_i<1/2$.
  
For a non equal-shift misreport, there are borrowers $q$ and $q'$, such that $p_{iq}-p_{iq'} = \hat{p}_{iq}-\hat{p}_{iq'} + \epsilon$, for  $\epsilon > 0$; i.e.,  with the relatively disadvantaged borrower labeled $q$. 
Consider a profile $p_{-i}$ that satisfies the following properties:

\begin{enumerate}
\item $K$ borrowers, including borrower $q$, are allocated:
\begin{itemize}
    \item For $q''\neq q$, set $p_{jq''}\in (p^*_{q''},1]$, for $j\neq i$, such that $B(q'')>1/2$, where this belief of others is feasible since $p^*_{q''}<1$.
    \item For borrower $q$, set $p_{jq}\in (p^*_q,\min(1,p^*_q+\frac{1}{\sum_{j'\neq i}w_{j'}}\frac{w_i\epsilon}{2}))$, so that $B(q)\in (\frac{1}{2},\frac{1}{2}+\frac{w_i\epsilon}{2})$, where this belief of others is feasible since $p^*_q<1$.
\end{itemize}
\item $K-m$ borrowers, including borrower $q'$, are not allocated.

\begin{itemize}
    \item For $q''\neq q'$, set $p_{jq''}\in [0, p^*_{q''})$, for $j\neq i$, such that $B(q'')<1/2$, where this belief of others is feasible since $p^*_{q''}>0$.
    \item For borrower $q'$, set $p_{jq'}\in (\max(0,p^*_{q'}-\frac{1}{\sum_{j'\neq i}w_{j'}}\frac{w_i\epsilon}{2})),p^*_{q'})$, so that $B(q')\in (\frac{1}{2}-\frac{w_i\epsilon}{2}, \frac{1}{2})$, where this belief of others is feasible since $p^*_{q'}>0$. 
    \end{itemize}
    \end{enumerate}

There is a non-zero measure on beliefs $p_{-i}$ satisfying these  properties by the full support assumption. 
%
For any such $p_{-i}$, at misreport $\hat{p}_i$  we have  $\hat{B}(q') > \hat{B}(q)$,  since  $\hat{B}(q') - B(q') = \hat{B}(q) - B(q) + w_i \epsilon$ and $B(q) - B(q') < w_i \epsilon$. By the monotonicity of the VCG allocation rule, this implies one of the following at this misreport:
\begin{enumerate}
    \item Borrower $q'$ but not $q$ is allocated, which is an outcome with lower total value  since $B(q')<B(q)$.
    \item Neither $q$ nor $q'$ are allocated, which is an outcome with lower total value since $B(q)>1/2$ and only $K-1$ other borrowers $q''$ have true aggregate belief $B(q'')>1/2$.
    \item Both $q$ and $q'$ are allocated, which is an outcome  with lower total value since $B(q')<1/2$ while $K$ borrowers $q''$ (including $q$) have aggregate belief $B(q'')>1/2$.  
\end{enumerate}

 This completes the proof.
$\blacksquare$
\end{proof}

We now state the second main theorem.
\begin{theorem}
\label{thm:2d}
 For possibly constrained liquidity $(K\leq m)$, one or more borrowers,  three or more recommenders, and a belief distribution with full support, the VCG Scoring mechanism with a lender profit threshold $c$, with $0<c<1$, is strict interim IC (strictly proper) when $\max_i [w_i] < \min(1-c, c)$ (which requires $n> 1/\min(1-c,c)$ recommenders).
\end{theorem}
\begin{proof}
We need three or more recommenders because  $\min(1-c,c)\leq 1/2$, and thus $n>1/(1/2)=2$.
We consider recommender $i$, belief $p_i$, any $\hat{p}_i\neq p_i$, and establish a non-zero measure on the beliefs $p_{-i}$ of others such that the allocation changes in a way that reduces the total value (i.e., not selecting the top borrowers amongst those with aggregate belief at least $c$).  Since VCG is weakly ex post IC (Theorem~\ref{vcg_dsic_1}), this establishes strict interim IC  up to equal-shift misreports.
 
For any $q$, let $B(q)$ denote the aggregate belief on $q$ at  $p_i$ and $\hat{B}(q)$ at report $\hat{p}_i$.
If $p_{-i}$ satisfies $p_{jq} = \frac{c - w_i p_{iq}}{\sum_{j' \neq i} w_{j'}}=p^*_q$,  $\forall j \neq i$, then $B(q) = c$. This belief $p_{-i}$ is feasible by full support, and since for $p_{iq}=0$ we have $p^*_q=c/\sum_{j'\neq i}w_{j'}<1$ since $\sum_{j'\neq i}w_{j'}>c$ from $w_i<1-c$. For $p_{iq}=1$, $p^*_q=(c-w_i)/\sum_{j'\neq i}w_{j'}>0$ since $w_i<c$.

For misreport $\hat{p}_i$, consider borrower $q$ with $\hat{p}_{iq}\neq p_{iq}$.

\textbf{(Case 1: $\hat{p}_{iq}<p_{iq}$)} Let $\hat{p}_{iq} = p_{iq} - \epsilon$, some $\epsilon>0$.
Consider a profile $p_{-i}$ that satisfies the following properties:
\begin{enumerate}
    \item $B(q)\in (c,c+w_i\epsilon)$, by setting $p_{jq}\in (p^*_q, \min(1,p^*_q+\frac{1}{\sum_{j'\neq i}w_{j'}}w_i\epsilon))$, all $j\neq i$, where this belief of others is feasible since $p^*_q<1$ and $c<1$.
    \item At least  $m-K \ (\geq 0)$ other borrowers $q'\neq q$ have $B(q')<c$, by setting $p_{jq'}\in [0,p^*_{q'})$, all $j\neq i$, where this belief of others is feasible since $p^*_{q'}>0$.
\end{enumerate}

There is a non-zero measure on beliefs $p_{-i}$ satisfying these properties by the full support assumption. Given (1) and (2), at true beliefs we have borrower $q$  allocated since $B(q)>c$ and at least $m-K$ others cannot be allocated,  so $q$ is in the top $K$ of those with aggregate belief above the threshold $c$. 
For any such $p_{-i}$, at misreport $\hat{p}_i$ we have $\hat{B}(q)=B(q)-w_i\epsilon<c$, since $B(q)\in (c,c+w_i\epsilon)$ and $\hat{p}_{iq} = p_{iq} - \epsilon$. This implies that $q$ is not allocated, resulting in an outcome with lower total value since $q$ was in the top $K$ and with true aggregate belief above the threshold.

\textbf{(Case 2: $\hat{p}_{iq}>p_{iq}$)}  Let $\hat{p}_{iq} = p_{iq} + \epsilon$, some $\epsilon>0$.

Consider a profile $p_{-i}$ that satisfies the following properties:
\begin{enumerate}
\item $B(q) \in (c - w_i \epsilon, c)$, by setting $p_{jq}\in (\max(0,p^*_q-\frac{1}{\sum_{j'\neq i}w_{j'}}w_i\epsilon),p^*_q)$, all $j\neq i$, where this belief of others is feasible since $p^*_q>0$ and $c>0$.
\item At least  $m-K \ (\geq 0)$ other borrowers $q'\neq q$ have $B(q')<c$, by setting $p_{jq'}\in [0,p^*_{q'})$, all $j\neq i$, where this belief of others is feasible since $p^*_{q'}>0$. 
\end{enumerate}

Given (1), at true beliefs  borrower $q$ is not allocated. At misreport $\hat{p}_i$,
we have $\hat{B}(q) = B(q)+w_i\epsilon> c$, since $B(q) \in (c - w_i \epsilon, c)$ and $\hat{p}_{iq} = p_{iq} + \epsilon$.
This implies one of the following at this misreport:
\begin{enumerate}
    \item Borrower $q$ is allocated, resulting in an outcome with lower total value since the true aggregate belief on $q$ is below the threshold (that is, by causing $q$ to be allocated, $i$ displaces a reserve borrower $q''$ with $B(q'') = c$, and $i$ must pay this difference to the system).
    \item If $q$ is not allocated, then since $\hat{B}(q)>c$ there must be $K$ others allocated, by the definition of the VCG outcome rule. At least $m-K$ others have $B(q')<c$, and thus at most $(m-1)-(m-K)=K-1$ others have $B(q')\geq c$. This means that at least one other borrower with $B(q')<c$ is allocated, and the outcome has lower total value.  
\end{enumerate}

This completes the proof.
$\blacksquare$
\end{proof}

\subsection{Strong Ex post IR}



We can also  achieve strong ex post IR by ensuring that the immediate payment {\bf by} each agent is weakly negative and noting that the outcome-contingent payments  {\bf to} each recommender are weakly-positive. 
For this, define $\mathit{tcomp}_i(\hat{p}_{-i})$ as the worst-case immediate payment in VCG given reports of others. This quantity is independent of the recommender's own report. At the same time,  we introduce a multiplier $\alpha>0$ to the outcome-contingent payments, so that $s_{iq}^{aVCG}(\hat{p}_{iq},o_q)=\alpha_iw_i$ if $o_q=1$, and 0 otherwise.  Neither change affects the incentive analysis.
%
%
Modifying the definition of reported valuations  accordingly, for example with $\hat{v}_i(a)=\sum_{q\in M}a_q \alpha w_i\hat{p}_{iq}$, we have
\begin{equation}
    \mathit{tcomp}_i(\hat{p}_{-i}) =  \mathrm{max}_{\hat{p}_i} \left( \sum_{j \neq i} \hat{v}_j(x_{-i}(\hat{p}_{-i})) - \sum_{j \neq i} \hat{v}_j(x^*(\hat{p}_i,\hat{p}_{-i})) \right).
\end{equation}

We refer to this as the {\em rescaled VCG scoring mechanism}. By {\em worst-case deficit} we mean the worst-case, total payment made by the mechanism to the agents, considering both the immediate payments and the outcome-contingent payments. 
\begin{theorem}
In the possibly constrained liquidity setting, and with multiple recommenders and multiple borrowers, there is some value of $\alpha_0>0$ such that for any $\alpha<\alpha_0$  the rescaled VCG scoring mechanism is strong ex-post IR, has worst-case deficit at most $\epsilon>0$, and strict proper incentives as stated in Theorems~\ref{thm:1d} and~\ref{thm:2d}. 
\end{theorem}
\begin{proof}
For  strong ex post IR, this follows  from the definition of $\mathit{tcomp}_i(\hat{p}_{-i})$ and outcome-contingent payments being non-negative. For the strict properness, this follows from the invariance of incentive analysis to scaling payments by any $\alpha>0$ and  that 
$\mathit{tcomp}_i(\hat{p}_{-i})$ is independent of recommender $i$'s reports.
The claim of deficit smaller than $\epsilon$ for any $\alpha<\alpha_0$, for some $\alpha_0>0$ follows from linearity, recognizing that $\alpha$ scales all payments.
\end{proof}

\subsection{Incentive Alignment with Better Reporting Quality}

We show in this section that a recommender  in the VCG scoring mechanism strictly improves its  utility by having a higher weight. This nicely aligns incentives in the broader ecosystem, similar to the way in which advertisers with higher quality make smaller payments in the generalized second pricing of internet advertising, and provides an incentive for a recommender to improve its reporting quality and thus attain a higher weight over time 
in the aggregation rule.
%
%
\begin{theorem} \label{increasing_weights}
Whatever the reports of others, for any recommender $i$, increasing the weight $w_i$ to $w_i'>w_i$, fixing the weights of others,  increases the  utility $U_i(p_i,p_i,\hat{p}_{-i})$ to the recommender from truthful participation in the VCG scoring mechanism. 
\end{theorem}
\begin{proof}
The utility to recommender $i$ is 
$$
U_i(p_i,p_i,\hat{p}_{-i}) = v_i(x^{VCG}(p))+ \sum_{j \neq i} \hat{v}_j(x^{VCG}(p_i,\hat{p}_{-i}))-\sum_{j \neq i} \hat{v}_j(x^{-i}(\hat{p}_{-i})), 
$$
where $x^{-i}$ is the  allocation decision that would be made in the VCG scoring mechanism without $i$.
Let $v_i$ and $v'_i$ denote the recommender's valuation for weight $w_i$ and $w'_i$, respectively.
The third term does not depend on its weight. Consider the first two terms, and let $a$ and $a'$ denote the allocation for $w_i$ and $w'_i>w_i$, respectively. We have $v'_i(a')+\sum_{j\neq i}\hat{v}_j(a')\geq v'_i(a)+\sum_{j\neq i}\hat{v}_j(a)> 
v_i(a)+\sum_{j\neq i}\hat{v}_j(a)$.  The first inequality holds trivially when $a'=a$, and if $a'\neq a$ then by the optimizing property of the VCG allocation rule. The second inequality holds  since  $v'_i(a) = \sum_q a_q w'_i p_{iq} > \sum_q a_q w_i p_{iq} = v_i(a)$, and since  reported values of others are unchanged.
$\blacksquare$
\end{proof}

\subsection{Relation to Chen et al.'s Impossibility Result}

Theorem 2 of~\citet{ChenYilingDMwG} states that a decision market, which uses belief reports as reflected in market prices to make a decision, is  strictly proper if and only if the decision is randomized and the distribution has full support. 
That is, for strict properness every decision must be taken with non-zero probability. 
The key difference between our model and that of \citet{ChenYilingDMwG} is that our agents report their beliefs simultaneously without awareness of the reports of others against whom they will be judged. This creates interim uncertainty about the allocation given the  common prior $\mathcal{D}$ and the technical conditions stated in Theorems~\ref{thm:1d} and~\ref{thm:2d}. Whereas~\citet{ChenYilingDMwG} have certainty about belief aggregation and thus need the decision rule to be randomized will full support, we have interim uncertainty about belief aggregation and the decision rule can be deterministic as a function of reports (still providing what is in effect full {\em interim} support) and thus strict properness.

\subsection{Linear Belief Aggregation}
As shown in theorems ~\ref{thm:1d} and~\ref{thm:2d} above, the VCG scoring mechanism is Interim IC when $max_{i \in N} (w_i)$ is sufficiently small. There are many linear aggregators in the literature which can be adjusted to satisfy this condition. In the Appendix, we describe the linear aggregator by \citep{BudescuDavidV2015IEtE}, which is based on reports from previous rounds and therefore doesn't affect current-round incentives.

\section{Conclusion} \label{conclusion_section}
Creditworthiness detection remains an unsolved problem, and our forumlation of the problem as an elicitation mechanism gathering community beliefs brings a novel source of information and a rich body of mechanism design literature to the task. In so doing, we developed a class of truncated, asymmetric scoring rules that are strict ex-post proper in the sufficient-liquidity case. We also connected the scoring rule and VCG literatures in a novel way, creating the \textit{VCG scoring mechanism} through which we construct agents' values via scoring rules. These values then feed into a VCG allocation and payment mechanism that is strictly interim IIC (and thus strictly proper) in both liquidity-constrained  and liquidity-unconstrained settings, and with or without a lender profit threshold.
 Given  impossibility results in the analagous case of decision scoring rules~\citep{ChenYilingDMwG}, these results expand the range of settings in which information can be elicited and paid for based on outcomes by carefully leveraging agents' interim uncertainty. We also connected  these mechanisms with the belief aggregation literature, allowing us to retain incentive compatibility properties with linear aggregation techniques.

Much future work remains. Lenders have suggested we incorporate considerations of collusion, where recommenders may misreport to help friends. Further work remains to ensure optimal aggregation. Bayesian and maximum likelihood estimator techniques could optimize the weighting of unknown agents and potential borrowers, incorporating factors such as demographics or other background information. Online-learning models could be built for specific use cases, such as optimizing the exploration and exploitation tradeoff for lenders building client bases in new communities. Also considering the required payments to motivate  recommenders to invest appropriately in providing good information is an  area that we have seen as important from an ongoing field study in Uganda. Finally, additional field work should continue in parallel to ensure that the theory focuses on the most critical problems which crop up in practice.

\bibliographystyle{plainnat}
\bibliography{refs}

\clearpage

\section[A]{Appendix}

\subsection{Brief Aggregation Literature Review}

 Linear  belief aggregators have been  shown to outperform individual experts in  a phenomena known as ``wisdom of the crowds"~\citep{MannesAlbertE2014TWoS}. Moreover, in many contexts expert quality will vary significantly, and aggregating expert reports based on quality  improves prediction accuracy \citep{SouleDavidAHfC, SatopaaVilleA2017PifM, MichaelD.Lee2014UCMt, Wang19}. We also make use of belief aggregation methods, which can be divided into linear \citep{BudescuDavidV2015IEtE,SouleDavidAHfC, WinklerRobertL1981CPDf} and non-linear \citep{RaykarVikasC2010LFC, Zhou, VenanziMatteo2014Cbam, WelinderPeter2010OcRa, KimJoonyoung2020RAMf, KargerDavidR2013Ecfm} approaches. They can be also be divided into approaches that weight based on previous prediction quality, e.g. \citet{BudescuDavidV2015IEtE} and ones that exploit structure in participants’ reports to provide ratings even before any outcomes have been realized, e.g. \citet{RaykarVikasC2010LFC}. Many of these latter approaches build on the expectation-maximization work of \citet{DawidA.P1979MLEo}. Any aggregation scheme that is a monotonically-increasing function of the reports will work for the truncated Winkler rule of Section \ref{winkler_section}. However, the aggregation must be a linear combination to work with the VCG Scoring mechanism of Section \ref{vcg_section}. As such, we focus here on linear aggregators that weight based on prediction quality for realized outcomes, where quality measures are derived from past performance.

%

\subsection{Proof of Theorem \ref{vcg_dsic_1}}
\begin{proof} Suppose by way of contradiction that recommender $i$ has a useful misreport $\hat{p}_i \neq p_i$ such that $\EX[u_i(\hat{p}_i, \hat{p}_{-i}, o)] > \EX[u_i(p_i, \hat{p}_{-i}, o)]$. Let $x'$ be the allocation under $i$'s misreport. Since $i$'s expected utility is strictly greater under $x'$ than $x$, we have

\begin{equation}
    v_i(x') - \sum_{j \neq i} \hat{v}_{j}(x^{-i}) + \sum_{j \neq i} \hat{v}_j(x') > v_i(x) - \sum_{j \neq i} \hat{v}_{j}(x^{-i}) + \sum_{j \neq i} \hat{v}_j(x)
\end{equation}

\begin{equation}
    \Rightarrow v_i(x') + \sum_{j \neq i} \hat{v}_j(x') > v_i(x^*) + \sum_{j \neq i} \hat{v}_j(x^*)
\end{equation}

But this is a contradiction, as the social choice function should then have allocated $x'$ to maximize recommender value. Note that this does not depend on the reports of others, whether they are being truthful, or even whether their reports are known. Thus, there can be no useful misreport, and the mechanism satisfies weak IC.

Efficiency, which is the property of a mechanism choosing the allocation that maximizes sum of value for all agents, follows from the fact that the allocation function explicitly maximizes the value of the aggregation function.

For IR, consider recommender $i$, consider the borrowers $x_{new}$ who are allocated 
when $i$ is present but not otherwise and $x_{old}$ who are not allocated when $i$ is  present but
 allocated otherwise. For other  borrowers, they contribute a non-negative utility to $i$, as
 they do not affect $t^{VCG}$ and the outcome-contingent payments are non-negative. 
 For the borrowers in $x^{new}$, expected utility is $\EX_{o \sim p_i}[u_i(p_i, \hat{p}_{-i}, o)] = v_i(x^{new}) - \sum_{j \neq i} v_j(x^{old}) + \sum_{j \neq i} v_j(x^{new})$. This term is also non-negative, because if it were negative, then the allocation function would have allocated to $x^{old}$ instead of $x^{new}$. Therefore, each recommender's expected utility is non-negative, and the mechanism is individually-rational.
$\blacksquare$
\end{proof}

\subsection{Linear Belief Aggregation Example}
Following the lead of~\citet{BudescuDavidV2015IEtE}, we calculate the weights to associate with a recommender based on the reports and outcomes from previous rounds.
First, we calculate a score $Q$ when considering the system of all recommenders 
 and the set $M^*$ of all borrowers from previous rounds that have received loans on which outcomes are available.
\begin{equation}
    Q = a + b \sum_{q\in M^*} \left( \sum_{r\in \{0,1\}} (o_{qr} - m_{qr})^2 \right),
\end{equation}
where  $m_{q1}$ is the {\em average (unweighted) reported belief that borrower $q$ will repay the loan 
across all recommenders} and $m_{q0} = 1 - m_{q1}$. It is convenient to set $a=100$ and $b=-50$, leading to $0\leq Q\leq 100$.
Define $Q_{-i}$ to be the analogous quantity, but where the average reported belief is calculated over those reports from recommenders $j\neq i$. Given this, we define the {\em accuracy contribution} $C_i=(Q-Q_{-i})/|M^*|$. 
%
%
The weight $w_i$ is defined to be zero for a recommender with $C_i\leq 0$. For the set of recommenders $N_+ = \{i\in N: C_i>0\}$, the weight $w_i$ for $i\in N_+$ is set in proportion to $C_i$ and normalized such that  $\sum_i w_i = 1$.

\end{document}